\documentclass[12pt]{article}
\usepackage{amsmath}
\usepackage{amssymb}
\usepackage{color}
\usepackage[left=3cm,top=3cm,right=3cm,bottom=3cm]{geometry}
\newtheorem{theorem}{Theorem}
\newtheorem{lemma}[theorem]{Lemma}
\newtheorem{corollary}[theorem]{Corollary}
\newtheorem{defn}[theorem]{Definition}
\newtheorem{observation}[theorem]{Observation}

\newtheorem{problem}{Open Problem}
\newtheorem{prop}[theorem]{Proposition}

\newcommand{\noin}{\noindent}
\usepackage{epsfig}
\newcommand{\qed}{\ \hfill \rule{1ex}{1ex}} %rightjustified box
\newenvironment{proof}{{\noin \bf Proof}: }{\qed \medskip}
\newcommand{\tell}{T_{\ell}}

\usepackage{tikz}
\usetikzlibrary{snakes,arrows}
\usetikzlibrary{calc}

\usepackage{graphicx}

\author{N.E. Clarke$^1$ \and D. Cox$^2$ \and  C. Duffy$^3$ \and D. Dyer$^4$ \and S. Fitzpatrick$^5$ \and  M.E. Messinger$^6$}

\title{Limited Visibility Cops and Robber}
\date{}
\begin{document}

\maketitle
\begin{small}
\begin{center}
	$^1$ Mathematics and Statistics, Acadia University, Canada\\
	$^2$ Mathematics, Mount Saint Vincent University, Canada\\
	$^3$ Mathematics and Statistics, University of Saskatchewan, Canada\\
	$^4$ Mathematics and Statistics, Memorial University of Newfoundland, Canada\\
	$^5$ Mathematical and Computational Sciences, University of Prince Edward Island, Canada\\
	$^6$ Mathematics and Computer Science, Mount Allison University, Canada\\
	\vspace{0.5in}
	August 23, 2017
\end{center}
\end{small}

\begin{abstract}
We consider a variation of the Cops and Robber game where the cops can only see the robber when the distance between them is at most a fixed parameter $\ell$.  We consider the basic consequences of this definition for some simple graph families, and show that this model is not monotonic, unlike common models where the robber is invisible.  We see that cops' strategy consists of a phase in which they need to ``see" the robber (move within distance $\ell$ of the robber), followed by a phase in which they capture the robber.  In some graphs the first phase is the most resource intensive phase (in terms of number of cops needed), while in other graphs, it is the second phase.  Finally, we characterize those trees for which $k$ cops are sufficient to guarantee capture of the robber for all $\ell \ge 1$.
\end{abstract}

\section{Introduction}\label{intro}
Cops and Robber is a well-known and well-studied pursuit-evasion game played on a graph.  The classical game is played on an undirected reflexive graph by two players: the first controls a set of cops and the second controls a single robber. The cops (any or all of which may move simultaneously) move so that some cop eventually captures the robber, while the robber moves to avoid the cops.  In this game, the cops and robber are always aware of the others' positions.

Cops and Robber was independently introduced by Nowakowski and Winkler \cite{NW} and Quillot \cite{Q}, with both papers characterizing those graphs for which one cop was sufficient to capture the robber. Much of the work surrounding this model has been involved with characterizing those graphs for which exactly $k$ cops are necessary and sufficient to capture the robber, with even the case $k=2$ being nontrivial \cite{CM12}. Beyond characterizations, there has also been work done on minimizing the length of the game, or the period for which the robber has been uncaptured \cite{BGHK09}.

Many variants of Cops and Robber have been considered. (See the book~\cite{Cops} for the extensive survey of Cops and Robber.) Some of these variants are due to restrictions on how the cops can move; for example, remaining adjacent to one another \cite{CN05}, or always moving along a shortest path to the robber \cite{Zombie1, Zombie}. Other recent work has included limiting the cops' knowledge of the robber's position.  In other words, exactly when the cops can see the robber is limited. When the robber is completely invisible, the resulting game looks a great deal like {\it edge-searching}, another pursuit-evasion game in which cops chase an invisible, fast-moving robber that may stop on vertices or edges. Edge-searching and node-searching are often used to give insight into the pathwidth of a graph, though the games themselves have spawned a great deal of research in their own respective rights. One variation is connected edge-searching, where the area cleaned by the searchers must induce a connected subgraph. (See, for example, \cite{YaDyAl09}.)

Loosely, the {\it zero-visibility cop number} of a graph $G$ is the minimum number of cops needed to guarantee the capture of an invisible robber, played under the standard Cops and Robber movement rules (so, unlike edge-searching, the robber moves at the same speed as the cops). This problem has been studied for a variety of graph families, and more recent work has been done on its complexity and its relation to pathwidth \cite{Tang,Tosic,Pathwidth,Complexity}. A variant where the robber is invisible but capture is not guaranteed, but expected to happen eventually, has also been studied \cite{KMP13}.

Of course, this problem naturally generalizes to the {\it $\ell$-visibility cop number}, where a cop can see any robber that is at distance at most $\ell$. This problem is an unusual hybrid of the classic Cops and Robber game and the zero-visibility version, as now the cops' task may be broken down into two sub-tasks: first, see the robber, and second, capture the robber.  The  {\it one-visibility} Cops and Robber game was considered  in~\cite{Tang}, and versions of the one-visibility  game where cops are allowed to ``jump'' as in edge-searching were investigated in \cite{AFGP16,bushcutting}.

The focus of this paper is the general $\ell$-visibility Cops and Robber game. We consider $\ell \geq 1$, unless explicitly stated otherwise.  In Section~\ref{sec:def} we formally introduce the Cops and Robber game, the zero-visibility Cops and Robber game, and the $\ell$-visibility Cops and Robber game. We also provide definitions for a number of concepts used throughout the paper.  In Section~\ref{sec:Prelim}, we observe the relationships between the $\ell$-visibility cop number for various values of $\ell$ and we  determine the $\ell$-visibility cop number for some  graph families.  Additionally, we comment on the relationship between the required number of cops on a graph and the number required on any retract, and on the role of isometric trees in computing the $\ell$-visibility cop number of a graph. Results in \cite{NW, BI, Tang}, among others, highlight the importance of distance preserving subgraphs and retracts in both the classical game and the zero-visibility game. This is also true in the $\ell$-visibility game, as demonstrated in Section \ref{sec:Prelim}.

Many pursuit-evasion games on graphs involve the concept of monotonicity: once a set of vertices is known to be ``robber-free'', then the cops can guarantee the set can remain ``robber-free" throughout the game.  Interestingly, and unlike edge-searching, we show the $\ell$-visibility Cops and Robber game is not monotonic. We demonstrate this fact in Section~\ref{sec:mono}.

In Section~\ref{sec:see}, we examine the relationship between the number of cops required to see the robber, and for which graphs this is a sufficient number to capture the robber.  We show for visibility at least two, that this is the case whenever the $\ell$-visibility cop number of a graph and
its (regular) cop number differ by at least two.  We also show this to be  true for all chordal graphs for any $\ell \ge 0$.

Finally, in Section~\ref{sec:trees}, we  provide a complete characterization for the $\ell$-visibility number of a tree, based on the structure of the tree.   We conclude with a series of open problems in Section~\ref{sec:open}.

% In the one-visibility game, once a cop has seen the robber, that cop can follow the path of the robber and at each round, the cop will see the robber's previous position.   This leads to another variant, called {\it time-delay} cops and robbers where, regardless of the vertices occupied by cops and the robber (i.e. regardless of the distance between cops and the robber), at each round the cops see the {\it previous} position of the robber. This variant is discussed in Section~\ref{sec:open}.

\section{Definitions and Preliminary results}\label{sec:tech}

\subsection{Definitions}\label{sec:def}
The classical game of Cops and Robber is played on an undirected reflexive graph with  two players: one player controls a set of cops and the other controls a single robber. (Typically, we will refer to the moves of the cops and robber, rather than of the two players.)  Before the game begins (considered to be on round 0), the cops choose  a set of vertices to occupy. The robber subsequently chooses  an unoccupied vertex to occupy.  At each round, the cops ``move'' and then the robber ``moves''.  In a ``move'' for the cops, each cop moves from their current vertex along an edge to a neighbouring vertex (the neighbouring vertex being the same as the initial vertex if a loop is traversed).  A move for the robber is defined similarly. Generally, we will neglect to consider loops, and instead allow some (or all) of the cops to remain on the same vertex on an unreflexive graph.  This is referred to as a {\it pass}.  The robber may also pass.

The original game is played with perfect information: both the cops and robber are aware of their opponents positions in each round. The cops win if, in some round, one cop moves to the vertex occupied by the robber; in such a situation we say the cops have {\it captured} the robber.  The robber wins if he avoids capture indefinitely.  It is assumed that both the cops and the robber play optimally; the cops to minimize the number of rounds in the game, and the robber to maximize.  The {\it cop number} of a graph $G$, denoted $c(G)$, is the minimum number of cops required to capture the robber in $G$.

The structural characterization of cop-win graphs is well known from~\cite{NW,Q}:   A graph $G$ is cop-win if and only if its vertices can be ordered $v_1, v_2, \ldots, v_n$ so that for each $v_i$, where $i >1$, there exists some $v_j$, where $j <i$, such that every $N[v_i] \cap \{v_1, \ldots, v_i\} \subseteq N[v_j]$.  This ordering is referred to as a {\it cop-win ordering} of $G$.  %We also say that $G$ is {\it dismantlable} whenever such an ordering exists. 
Let $H_{i}$ be the subgraph induced on $\{v_1, v_2,  \ldots, v_i\}$.  The vertex $v_i$ is referred to as a {\it corner} in $H_i$ and  vertex $v_j$ is said to {\it dominate} $v_i$ in $H_i$, whenever $N_{H_i}[v_i] \subseteq N_{H_i}[v_j]$.   It follows that $v_n$ is a corner in $G$, and $H_n = G$.

%The game dramatically changes if the cops visibility of the robber is restricted.  As a simple example, consider a finite %tree $T$.  It is easy to see that $c(T)=1$.  At each round, the cop moves along the unique shortest path to the robber %and, as $T$ is finite, the cop will eventually (i.e. in a finite number of rounds) capture the robber.  On the opposite end %of the visibility spectrum, suppose the robber is invisible.  The cops can see other cops' positions, but cannot see the %position of the robber (unless capture has occurred).  However, the robber will still be able to see the positions of the %cops.  The {\it zero-visibility cop number} of a graph to be the minimum number of cops required to guarantee capture %of an invisible robber in a finite number of rounds.  Certainly one cop is no longer sufficient to capture the robber on %some trees: as an example, consider the graph obtained by merging an endpoint of $3$ copies of $P_5$; in such a %graph, two cops are necessary.

The rules of the  $\ell$-visibility game, where $\ell \ge 0$,  do not vary from the classical game when it comes to both choice of positions in round 0 and the movement of the players.   The condition for winning is also the same in both versions of the game.  However, in the $\ell$-visibility game, the cops do not have perfect information, while the robber does.   In the $\ell$-visibility game, we will say that the cops \emph{see} the robber if any cop and the robber simultaneously occupy some vertices $x$ and $y$, respectively,  such that $d(x,y) \le \ell$.  (You could think of the cop within distance $\ell$ seeing the robber and sharing that information with the other cops.)   We note  that for a graph $G$ if $\ell \geq diam(G)$, then $\ell$-visibility Cops and Robber is identical to the classical game.

Let $c_\ell'(G)$ and $c_\ell(G)$ denote the minimum number of cops required to {\it see} and {\it capture} (respectively) the robber in the $\ell$-visibility Cops and Robber game.  It was noted in~\cite{Tang} that $c_1'(G) \leq c_1(G) \leq c_0(G)$ for all graphs $G$.  Note that {\it seeing} the robber does not necessarily imply capture of the robber: $c_1'(C_4) = 1 < 2 = c_1(C_4)$.  However, it does imply capture on chordal graphs (see Theorem~\ref{thm:ellchord}).

%We may initially view this game as a graph searching problem in which the role of the cops is to clean the vertices. In %such a formulation, initially every vertex is marked as \emph{dirty}; equivalently, this is the {\em robber territory}, all the %possible locations of the robber. At each of the cops' turns \slf{is this before or after the cops move?}, the closed %neighbourhood of the new location of the cop is marked as \emph{clean}. After each of the robber's turns, every clean %vertex that is not in the closed neighbourhood of any cop that has a dirty neighbour is marked as dirty. However, while %this alternate paradigm works well for zero-visibility games, it does break down; once the robber has been seen by any %cop, the dirty vertices become extremely curtailed.

To properly describe the strategy for a finite set of cops, we require the following notation.
Let $G$ be a  connected graph. Let $\mathcal{L} = \{\ell_i\}_{i=1}^{i=k}$, where $\ell_i = \ell_i(0), \ell_i(1), \dots$ is a walk that describes the position of cop $i$ in $G$, where the argument indicates the round. We call $\mathcal{L}$ a \emph{$k$-cop strategy}.
%We say that $\mathcal{L}$ is \emph{successful} if for all walks $W = w(0), w(1), \dots$ there exist $t \geq 0$ and $1 \leq i \leq k$ such that $w(t) \sim  \ell_i(t)$ or $w(t)  \sim \ell_i(t+1)$. It is easy to see that a $k$-cop strategy is successful if and only if the robber is captured. As such $c_\ell(G)$ is the least integer $k$ such that there exists a successful $k$-cop strategy.
We say that a cop \emph{vibrates} between time $t$ and $t+2k$ if there exists $uv \in E(G)$ such that $$\left(\ell_i(t), \ell_i(t+1), \dots, \ell_i(t+2k-1), \ell_i(t+2k)\right) = \left(u,v, \dots, u,v,u\right).$$

As the cops proceed through their strategies, the set of vertices that may possibly contain the robber decrease. Borrowing from the language of edge searching, a vertex known to not contain the robber is called \emph{clean}, otherwise it is \emph{dirty}. A vertex that is clean but again becomes dirty at a later time is said to be \emph{recontaminated}. A set of cops \emph{clean} a subgraph by adding each vertex of the subgraph to the set of clean vertices while allowing no vertex of the subgraph to become recontaminated. We refer to the set of dirty vertices as the \emph{robber territory}. We say that a $k$-cop strategy is a \emph{$k$-cop cleaning strategy} or a \emph{cleaning strategy using $k$ cops} if there is some time $t$ so that all of the vertices are clean at time $t$; that is, a time $t$ when the robber can be guaranteed captured.

 A subgraph $H$ of $G$ is {\it isometric} if for any pair of vertices $x,y \in V(H)$, $d_H(x,y) = d_G(x,y)$. In the case that $H$ is a tree, we say that $H$ is an isometric tree in $G$.  Recall that a graph homomorphism is a vertex mapping such that adjacency is preserved.  Given a (reflexive) graph $G$ and a subgraph $H$, we say that $H$ is a {\it retract} of $G$ if there is a homomorphism $f: V(G) \rightarrow V(H)$ such that for every $v \in V(H)$, $f(v) = v$.  We refer to the mapping $f$ as a {\it retraction} of $G$ onto $H$.

Throughout, we assume that graphs are finite and simple and contain a single robber. We note that there is a rich compendium of results for Cops and Robber on infinite graphs (for example see \cite{infinite}). We discuss the problem of limited visibility Cops and Robber on an infinite graph in Section  \ref{sec:open}. As we do not explore the concept of capture time and are restricted to finite graphs, our assumption that the graph contains a single robber is well-founded, as a strategy that is utilized to capture a single robber can be repeated until no robber remains.

\subsection{Preliminary Results}\label{sec:Prelim}

The notion of retract described above plays an important role in the classification of cop-win graphs \cite{NW}. We note that if $H$ is a retract in $G$, then $H$ is also an isometric subgraph in $G$. This follows from the fact that a retraction onto $H$ is edge-preserving on $G$, and the identity on $H$. We prove that for any retract $H$ of $G$, $c_\ell(H) \le c_\ell(G)$.  The proof of this result is similar for that of the cop number of retracts: $c(H) \le c(G)$ \cite{BI}.
	
	\begin{theorem} \label{thm:retract}
		For any retract $H$ of a graph $G$, $c_\ell(H) \le c_\ell(G)$ and $c_\ell'(H) \le c_\ell'(G)$ .
	\end{theorem}
	
	\begin{proof}
		Suppose $H$ is a retract of $G$ with retraction $f$, and $c_\ell (G)=k$.  We show $c_\ell (H)\leq k$.
		
		We consider a pair of $\ell$-visibility games, each with $k$ cops and one robber, played in parallel.  The first game is played on $G$ and the second game is played on $H$.  For the robber, the moves in the two games will be identical. (One can think of the robber as being totally unaware of the first game, and making decisions to avoid capture as long as possible in the second game only.)  As for the cops, in the first game, each of the $k$ cops will play according to a winning strategy in $G$.  In the second game, the cops will use the moves in the first game and the retraction $f$ to determine their moves.  Specifically, for each cop $C$ in $G$, there is a corresponding cop  $C'$ in $H$.  Whenever $C$ moves onto a vertex $v \in V(G)$, $C'$ will move onto the vertex $f(v)$ in $H$.  This is always possible since $f$ is edge-preserving.
		
		In order for the cops in $H$ to adhere to the rules of the $\ell$-visibility game, the cops in $G$ cannot have information regarding the robber's position that is unavailable to the cops in $H$.  To verify that this is the case, consider $x \in V(G)$ and $y \in V(H)$ such that $d_G(x,y) \le \ell$.  It follows that $d_H(f(x), f(y)) \le
		d_G(x,y) \le \ell$.  Since $f(y) = y$, we have $d_H(f(x), y) \le \ell$.  Therefore, whenever a cop $C$ sees the robber in $G$, the corresponding cop, $C'$, in $H$ also sees the robber.
		
		Since the robber is restricted to $H$ in both games, he will be captured in the first game on a vertex of $H$.  Therefore, in the second game, he will also be captured in $H$.  Therefore, $c_\ell (H) \le c_\ell (G)$.  It can be similarly shown that $c'_\ell(H) \le c'_\ell (G)$.
	\end{proof}
	
	It was shown in \cite{NR} that any isometric tree in $G$ is also a retract in $G$.  We, therefore, have the following corollary regarding isometric trees.
	
	\begin{corollary} \label{cor:isometrictree}
		If $T$ is an isometric tree in $G$ then $c_\ell (T) \le c_\ell (G)$.
	\end{corollary}

For any $k \ge 1$, a $k$-dominating set of $G$  is a subset $S$ of $V(G)$ such that every vertex in $V(G) \setminus S$ is at most distance $k$ from some vertex of $S$.  The $k$-domination number of $G$, denoted $\gamma_k(G)$, is the minimum cardinality of a $k$-dominating set of $G$.  When $k=1$, this is simply referred to as the domination number of $G$.  We begin by stating some obvious, but useful inequalities.

\begin{prop}\label{prop:domSet}
	For any connected graph $G$,
\begin{enumerate}
\item  $c'_\ell(G) \le c_\ell(G)$;
\item $c'_\ell (G) \le \gamma_\ell (G)$;
\item $c_0 (G) \ge c_1(G) \ge c_2 (G) \ge \cdots \ge c_{diam(G)}(G) = c(G)$;
\item $c'_0 (G) \ge c'_1(G) \ge c'_2 (G) \ge \cdots \ge c'_{rad(G)}(G) =1$.
\end{enumerate}
\end{prop}

We next state some easy results for paths, cycles, complete bipartite graphs and complete graphs.  On a cycle $C_n$ with $n \geq 2\ell+2$, an $\ell$-visibility cop chooses a vertex to initially occupy and can see a total of $2\ell+1$ vertices (including the vertex occupied by the cop).  The cop first moves to an adjacent vertex and can now see a vertex that the cop could not previously see.  As there is a delay between the cops move and the robber's, the cop can effectively see $2\ell+2$ vertices of the graph during any given round.

\begin{prop}\label{cor:complete} For $\ell \geq 1$,
	\begin{enumerate}
		\item $c_\ell(P_n) = c_\ell'(P_n) = 1$ for any $n \geq 1$;
		\item $c_\ell(C_n) =  2$ for any $n \geq 4$;
		\item $c_\ell'(C_n) = \begin{cases} 2 & \text{ if } n \geq 2\ell+3, \\ 1 & \text{ otherwise;}\end{cases}$
		\item $c_\ell(K_{m,n}) = 2$ for any $2 \le m \le n$;
		\item $c_\ell'(K_{m,n}) = 1$ for any $2 \leq m \leq n$;
		\item $c_\ell(K_n) = c_\ell'(K_n) = 1$ for any $n \ge 1$.		
	\end{enumerate}
\end{prop}
There are two cases of  particular note in Proposition~\ref{cor:complete}. While it has already been noted that seeing is not the same as capturing (as in the graph $C_4$), part 3 of Proposition~\ref{cor:complete} goes a step further. Consider the cycle $C_5$ in the $1$-visibility game.  The cop chooses a vertex to occupy and can ``see" a total of $3$ vertices (including the vertex he occupies).  In his first move, he moves to an adjacent vertex.  He can now see a vertex he could not previously see.  If the robber is not there, then the cop does not see the robber, but now the set of vertices that can possibly contain the robber, that is, the \emph{robber territory}, is only a single vertex. That is, the cop now knows exactly where the robber is located! In this case, ``locating'' is stronger than seeing.

It is also interesting to note that $c_\ell(G) = 1$ for any graph $G$ with a universal vertex. The graph $K_n$ is merely a classic example, and particularly interesting because $c_0(K_n) = \lceil  n/2 \rceil$.

%[ADD MORE OBVIOUS GRAPHS FOR c and c'] \me{$\checkmark$}

%\me{Note:  Something to think about with respect to our ``seeing'' definition.  Consider a cycle of length $2\ell+3$.  The cop chooses a vertex to occupy and can ``see" a total of $2\ell+1$ vertices (including the vertex he occupies).  In his first move, he moves to an adjacent vertex.  He can now see a vertex he could not previously see.  If the robber is not there, then the cop does not see the robber, but knows exactly where the robber is located before the robber moves (after the robber moves, the cop knows he's on one of two possible vertices).  So after the cop moved but before the robber moved, the cop knows exactly the vertex occupied by the robber even though he has not seen it.  Thus, as the definition is currently stated, this implies $c_\ell'(C_n) = 2$ if $n = 2\ell+3$.}

%\cd{This has always bothered me. As stated there is a difference between seeing the robber and deducing where the robber is. I don't know what the resolution is. I am drawn to the idea of dumb cops -- they just follow their walk until they either see the robber or reach the end of the walk. Are there other cases where this occurs, aside from the cycle?}

%\dd{Don't know if there's others, but I put ME's discussion in.}

%Using Proposition \ref{prop:domSet} we give an illustrative example of a graph for which $c_1(G) \neq c_1^\prime(G)$.

In Section \ref{sec:see}, we investigate the relationship between $c_\ell'(G)$ and $c_\ell (G)$, and show that for cop-win graphs $c_\ell(G) - c'_\ell (G)$ is at most one.  In Section \ref{sec:trees} we see that  the difference between $c_\ell (G)$ and $c(G)$ can be arbitrarily large.  We now show that difference between both $\gamma(G) -c'_1(G) $ and $\gamma (G) - c_1(G)$ can be arbitrarily large.

In Figure~\ref{fig:Ex1}, we have a graph $G$ such that $c(G) = 1$, as a cop-win ordering is given by the order of the vertices. By placing a single cop on $v_3$, the cop can see every vertex except $v_5$ and $v_6$. If the robber has not already been seen, the cop's first move is to $v_7$, and the robber is seen. Thus, $c_1'(G) = 1$.  Since there is no dominating vertex, it is straightforward to see that $\gamma(G) = 2$. However, the next result shows that $c_1(G) = 2$, and therefore, seeing does not imply capture even on a cop-win graph.

\begin{figure}[htbp]		
	\[ \includegraphics[width=0.35\textwidth]{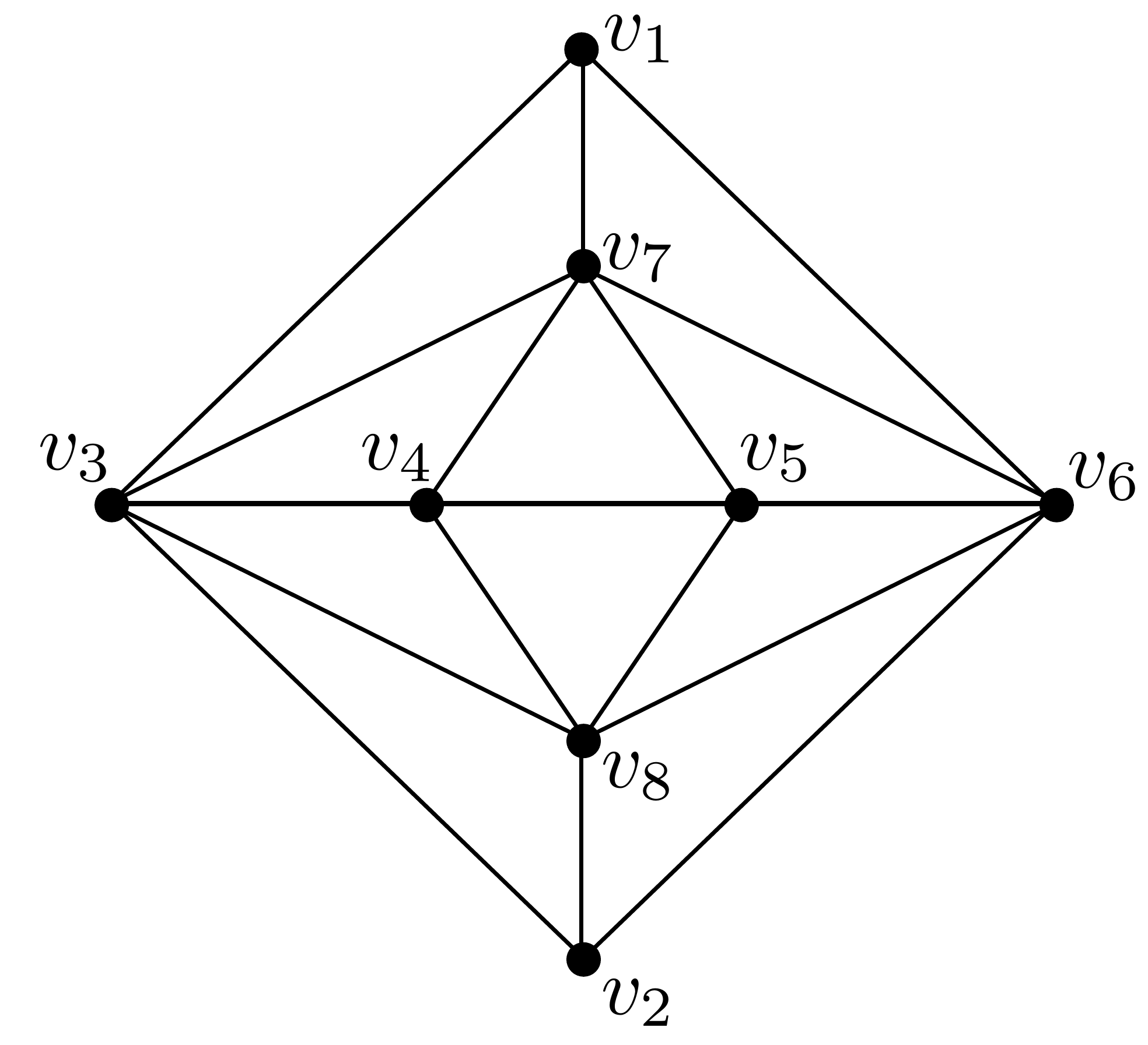} \]
	\caption{A graph $G$ in which $c(G)=1$, $c_1'(G) =1$, and $c_1(G) = 2$.}		
	\label{fig:Ex1}
	
\end{figure}

\begin{observation} \label{obs:Ex1} For the graph $G$ given in Figure~\ref{fig:Ex1}, $c_1(G)=2$. \end{observation}

\begin{proof} Let $G$ be the graph shown in Figure~\ref{fig:Ex1} and for a contradiction, suppose $c_1(G)=1$.  Consider the final move for the cop:  the robber must have occupied a vertex $x$ and the cop occupied a vertex $y$ such that $N[x] \subseteq N[y]$ (i.e. $x$ was a corner).  Then any vertex to which the robber can move, the cop can also move.  (If $N[x] \not\subseteq N[y]$, then the robber could have moved to a vertex $z \notin N[y]$ and avoided capture in the subsequent round.)  Note that $G$ has two corners, $v_1$ and $v_2$.  Without loss of generality, assume the robber occupied $v_1$ immediately prior to capture and the cop occupied $v_7$.  In the previous round, the cop must have occupied $v_4$ or $v_5$ before his penultimate move to $v_7$ (if the cop had occupied $v_3$ or $v_6$, then he could have moved to $v_1$ instead of $v_7$ and captured the robber earlier).  Suppose the cop occupied $v_4$ before moving to $v_7$.  After the cop  moved to $v_4$, the robber moved to $v_1$.  In the previous round, the robber must have occupied a vertex of $N[v_1]$.  If the robber had occupied $v_3$ or $v_6$ then he could have moved to either $v_1$ or $v_2$ and the cop would not know to which vertex the robber had moved.  Consequently, the cop would not know whether to move to $v_7$ or to $v_8$ in his penultimate move, which yields a contradiction. If the robber occupied $v_1$ or $v_7$, then he could have moved to $v_6$ instead of $v_1$.  In this case, even assuming the cop could see that the robber moved to $v_6$, the cop cannot capture the robber in two moves, which is a contradiction.  Hence, $c_1(G) >1$.
	
Since $\gamma(G)=2$, two 1-visibility cops can occupy dominating vertices of $G$ and capture the robber immediately.  \end{proof}

We next provide a family of graphs that shows that the difference between $\gamma$ and $c_1$ can be arbitrarily large.  We begin with $k$ copies of the graph $G$ given in Figure~\ref{fig:Ex1}, for finite $k$.  Let $G_k$ be the graph obtained by merging vertices of degree $5$ from $k$ copies of $G$ as shown in Figure~\ref{fig:Ex2}.

\begin{figure}[htbp]
	
	\[ \includegraphics[width=0.8\textwidth]{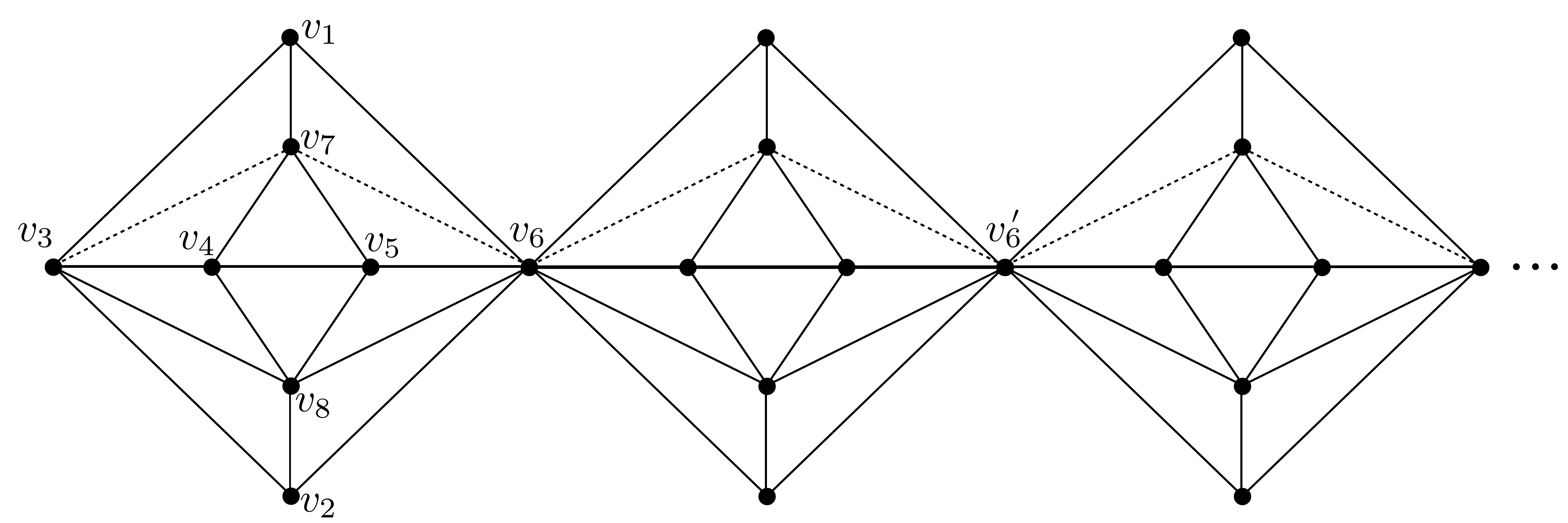} \]
	
	\caption{A graph $G_k$ for which $\gamma(G_k) = k+1$, but $c'_1(G_k)=1$ and $c_1(G_k)=2$.}
	
	\label{fig:Ex2}
	
\end{figure}

Clearly $\gamma(G_k)=k+1$.  However, we can easily observe that $c_1'(G_k) = 1$.  Initially place a $1$-visibility cop at an endpoint of the path indicated with dotted edges (labeled as $v_3$ in Figure~\ref{fig:Ex2}).  If the cop does not see the robber, then the cop moves along the path indicated by dotted edges.  If the robber was initially located at $v_5$ or $v_6$, then the cop will see the robber.  Otherwise, the robber was not initially in the leftmost ``copy" of $G$.  A similar argument shows that if the cop continues to move along the path indicated by dotted edges, he will eventually see the robber.  With the addition of a second 1-visibility cop, this strategy can easily capture the robber: a 1-visibility cop occupies $v_6$ while the other occupies $v_3$.  If the robber is located in the leftmost ``copy'' of $G$, he is captured immediately.  Otherwise, one cop moves from $v_3$ to $v_2$ and then to $v_6$.  After both cops occupy $v_6$, one cop moves to $v_6'$ (labeled in Figure~\ref{fig:Ex2}) and the cops follow the same strategy to clean the next ``copy'' of $G$.   The graph $G$ is a retract of $G_k$, since every unlabelled can be mapped to $v_6$, while every labelled vertex can be  mapped to itself. By Theorem~\ref{thm:retract}, we know $c_1(G_k) \geq c_1(G) = 2$.  Combined with the above argument, this yields $c_1(G_k)=2$.

Since $c_1(K_n)=1$  and there exist trees with arbitrarily large 1-visibility cop number, the 1-visibility cop number is not closed under minors in general.  However graph $G_k$ does highlight the existence of a minor relationship, based on cut-vertices, cut-edges, and isometric subgraphs.

\begin{observation}  Let $v$ be a cut-vertex or an endpoint of a cut-edge in graph $G$.  For any $\ell \ge 1$,  we have $\displaystyle  c_\ell(G) \leq 1+ \max_{1 \leq i \leq \alpha} \{c_\ell (H_i)\}$, where $H_1,H_2,\dots, H_\alpha$ are the subgraphs induced by the deletion of $N_\ell[v]$ from $G$.
\end{observation}

If $v$ is a cut-vertex (or endpoint of a cut-edge) in a graph $G$, then a 1-visibility cop can occupy $v$.   The robber is then restricted to some subgraph $H_i$.  The remaining $\max_{i \in [\alpha]} \{c_1(H_i)\}$ cops then search each subgraph $H_1,H_2,\dots, H_\alpha$ one by one until the robber is caught.

\section{Monotonicity}\label{sec:mono}

Many cops and robber games involve the idea of monotonicity. Essentially a strategy is monotonic if, once a portion of the graph is known to be free of the robber, then the cops have a strategy to guarantee that the robber cannot re-enter such a portion. Edge-searching requires no extra cops to guarantee a monotonic strategy; alternatively, we say that edge-searching is monotonic~\cite{LaPa93, BiSe91}.  However, connected edge-searching requires strictly more cops to maintain monotonicity; that is, connected edge-searching is not monotonic~\cite{YaDyAl09}.

In the classical Cops and Robber game, the robber is visible, and so the idea of monotonicity has not been of interest to researchers.  However, in the zero-visibility Cops and Robber game, which conceptually is quite similar to edge-searching, monotonicity becomes a natural concern. It was shown in \cite{Pathwidth} that the zero-visibility cop and robber game is clearly not {\em strictly} monotonic (due to the usefulness of vibrating, as described in Section~\ref{intro}).  It is also not {\em weakly} monotonic -- even allowing the recontamination associated with vibrations, more cops are required to show that the size of the robber territory (considered immediately after the cops' move) is always non-increasing over time.

Certainly, there are similar considerations in $\ell$-visibility Cops and Robber. Let $S_k$ denote the robber territory immediately after the cops have taken their $k$th move, but before the robber moves. We define the {\em (weakly) monotonic $\ell$-visibility cop number} of a graph $G$,  $mc_\ell(G)$, to be the minimum number of cops required to guarantee capture of the robber in $G$ with the restriction that $S_{k+1} \subseteq S_{k}$ for all moves after the initial placement.

Consider a perfect binary tree $T$ of depth $3$. Let every edge of $T$ be subdivided $2\ell+1$ times; call the resulting tree $\tell$. Let the root of $\tell$ be $r$, the ``left'' vertex of depth $2\ell +2$ be $a$, the right descendent of $a$ that is $2\ell+2$ levels below be $b$, and the vertices on the path from $a$ to $b$ be labelled $v_1$, $v_2$, \ldots, $v_{2\ell+1}$, with $v_1$ adjacent to $a$ and $v_{2\ell+1}$ adjacent to $b$, as shown in Figure~\ref{monotonicfigure}.

\begin{figure}[ht]
\begin{center}
\includegraphics[scale=0.365]{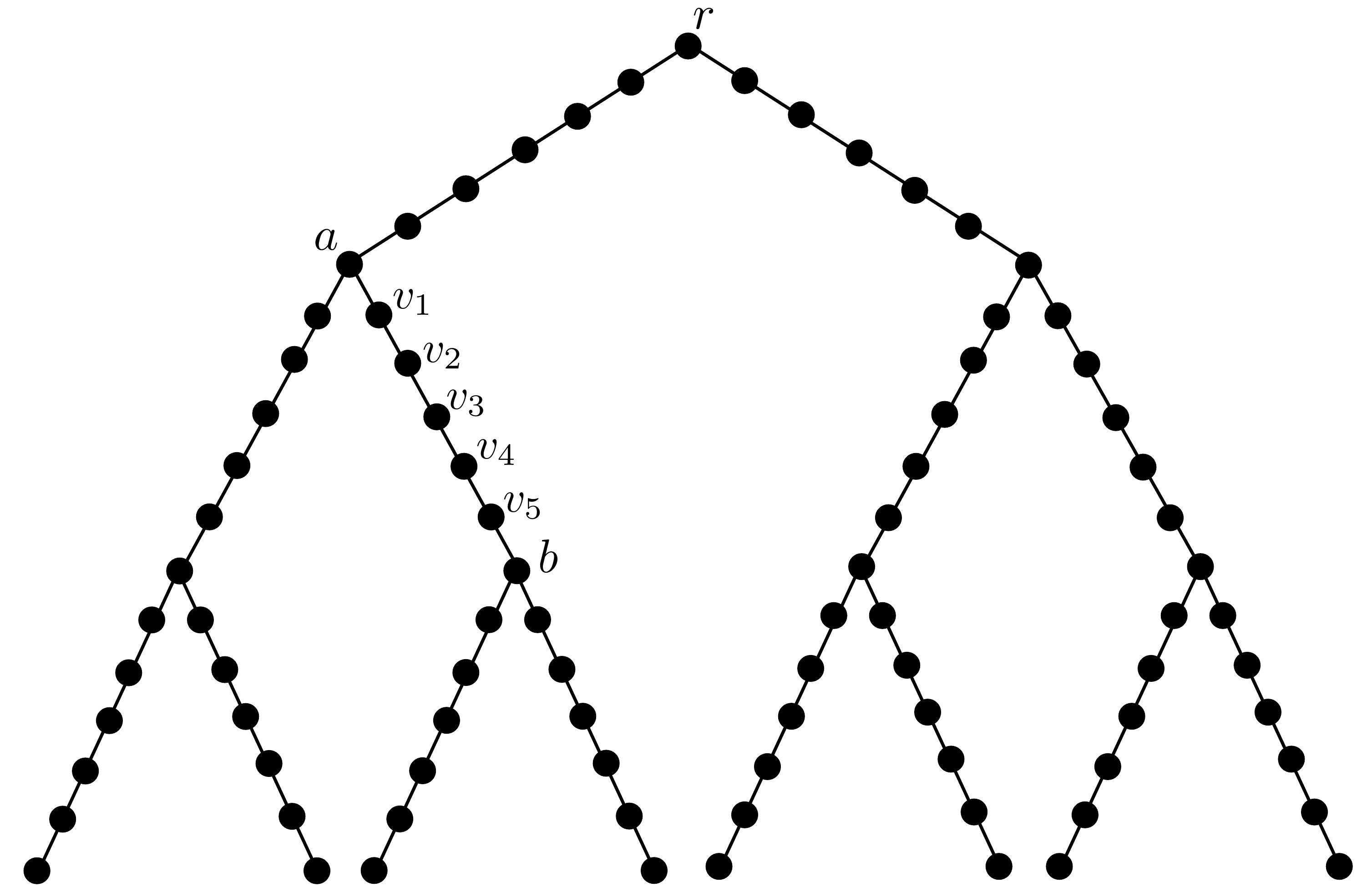}
\end{center}
\caption{The tree $T_\ell$ when $\ell = 2$.}

\label{monotonicfigure}
\end{figure}

\begin{theorem} For $\ell \geq 1$, $c_\ell(\tell) = 2$, and $mc_\ell(\tell)=3$.
\end{theorem}

\begin{proof}
One $\ell$-visibility cop is not sufficient to guarantee capture of the robber (see Section \ref{sec:trees} for further detail).  We present a (non-monotonic) strategy for two cops, $C_1$ and $C_2$, to capture the robber.

First, place one cop on each of the leaves to the left of vertex $a$. They may both move up the tree, eventually both ending at $a$. Then, cop $C_1$ moves from $a$ to $v_{\ell}$, and in subsequent moves will vibrate between $v_{\ell}$ and $v_{\ell+1}$ until all vertices below $b$ are robber-free. Thus, $C_1$ will guard the branch of $\tell$ to the left of $a$, as in every second move, $a$ will be surveilled; that is, $C_1$ will see the robber and be able to guarantee capture, by Theorem~\ref{thm:ellchord}.  Then cop $C_2$ may move to the left descendant of $b$ that is $\ell$ levels above a leaf, and proceed to the corresponding right descendant of $b$. This cleans all the descendants of $a$, at which point both $C_1$ and $C_2$ may proceed to $r$, and clean the right half of $\tell$ in the reverse of this strategy. Thus, $c_\ell(\tell) = 2$.

However, this is not a monotonic strategy; the vertex $b$ is known to be robber-free ($C_2$ even occupies it), but after that, $b$ is unseen for three successive moves, allowing recontamination. This may be repaired by introducing a third cop, $C_3$. After cleaning the left branch of $a$, while $C_1$ vibrates between $v_{\ell}$ and $v_{\ell+1}$, $C_3$ may move to vertex $b$, while $C_2$ proceeds as in the previous strategy; again after cleaning both branches of $a$, all cops may move to $r$ and repeat the reverse of the strategy. This shows that $mc_\ell(\tell) \leq 3$, and it is straightforward to show that there is no $2$-cop strategy that preserves monotonicity.
\end{proof}

\section{Seeing vs Capturing}\label{sec:see}

We now attempt to answer the question ``For which graphs does seeing imply capture?"   We have seen in Proposition \ref{cor:complete} that seeing does not imply capture for complete bipartite graphs and some cycles.  However, in these examples,  the cops could see almost the entire graph from fixed starting positions.  In general, this will not be the case, and the searching phase of the cops' strategy may be the phase that is the most resource intensive.

First, we investigate this problem for the $\ell$-visibility game when $\ell \ge 2$.  We then examine the problem on chordal and cop-win graphs.

\begin{theorem}\label{thm:see}
For any graph $G$ and $\ell \ge 2$, either $c_\ell'(G) = c_\ell (G)$ or $c(G) \le c_\ell (G) \le c(G) + 1$.
\end{theorem}
\begin{proof}
Consider a graph $G$ and a set of cops of size $m$ where $m = \max \{c'_\ell (G), c(G) +1\}$ and $\ell \ge 2$.  Since $m \ge c'_\ell(G)$, there is a strategy on $G$ for the $m$ cops to eventually see the robber.  Therefore, at some point in the game, a cop $C$ is within distance $\ell$ of the robber.  Suppose this occurs in round $t'$.

In any given round, if the robber occupies vertex $r$ at the end of the round, we let $r'$ denote the vertex occupied by the robber at the end of the previous round.  If we associate the robber with vertex $r$, then we refer to $r'$ as the robber's shadow.    Starting in round $t'+1$,  we assume that $C$ moves on a geodesic between himself and the robber's shadow until indicated otherwise.  As a result, $C$  maintains a distance of at most $\ell$ from the robber's shadow at all times.  Now, the other $m-1$ cops have perfect information as to the position of the robber's shadow.  Since $m-1 \ge c(G)$, these $m-1$ cops can capture the robber's shadow.

Since the robber and his shadow are distance at most one from each other the robber has either been captured, or there is a cop $C'$ adjacent to the robber.  Assume the latter.  Now, $C'$ will continue to move onto the robber's shadow in each round,  and $C$ will abandon his previous strategy and join the other $m-2$ cops.  From this point on, $C'$ is always withing distance two of the robber, and therefore, the other $m-1$ cops have perfect information as to the position on the robber.  Since $m-1 \ge c(G)$, these $m-1$ cops can now capture the robber. Hence, $c_\ell (G) \le  m$.

If $m = c(G) +1$, then $c(G) \le c_\ell (G) \le c(G)+1$.  If $m = c_\ell (G)$, then  $c'_\ell (G) \le c_\ell (G) \le c'_\ell (G)$.  The result follows.
\end{proof}

It is well-known that chordal graphs provide a large family of graphs $G$ for which $c(G) = 1$ \cite{NW}. However it is easy to construct examples of chordal graphs that require at least two cops to capture the robber with limited visibility (see Theorem \ref{thm:treek}). However, in this section we show that for $\ell$-visibility Cops and Robber, chordal graphs do have the property that if at any point there is a cop at distance no more than $\ell$ from the robber, then that cop can eventually capture the robber. We begin by stating several useful definitions from~\cite{west}.  A vertex of graph $G$ is {\it simplicial} if its neighbourhood in $G$ induces a clique.  A {\it simplicial elimination ordering} is an ordering $v_n,\dots,v_1$ for deletion of vertices so that each vertex $v_i$ is a simplicial vertex of the remaining graph induced by $\{v_1,\dots,v_i\}$.  It is well-known that $G$ has a simplicial elimination ordering if and only if $G$ is chordal.

	In the following proof, using zero-visibility as a base case, we show that if a cop sees the robber, then in a finite number of subsequent rounds the distance between that cop and the robber decreases.
		
	\begin{theorem}\label{thm:ellchord}
		If $G$ is a chordal graph and  $\ell \geq 0$, then $c_\ell'(G) = c_\ell(G)$. Furthermore, once the robber has been seen by a cop, that cop can capture the robber.
	\end{theorem}
	
	\begin{proof} Let $G$ be a chordal graph and fix a simplicial elimination ordering of the vertices: $v_n, v_{n-1}, \dots, v_1$.  We proceed by inducting on $\ell$. As the case where $\ell=0$ is trivial, we suppose $c_{\ell-1}(G) = c'_{\ell-1}(G)$ for some $\ell \geq 1$. Furthermore, we suppose once the robber has been seen by an $(\ell-1)$-visibility cop, that cop can capture the robber.

Initially, $c_\ell'(G)$-many $\ell$-visibility cops follow a strategy that allows an $\ell$-visibility cop $C$ to see the robber. Thus, at some point, $C$ occupies a vertex that is distance $\ell$ from the vertex occupied by the robber, $R$; suppose they are located on vertices $v_i$ and $v_x$, respectively. We assume it is the robber's turn to move, as otherwise that cop will move to a vertex that is distance $\ell-1$ from the robber (if $\ell =1$ then the cop occupies the same vertex as the robber) and by the induction hypothesis, $C$ can capture the robber.

Observe that the robber cannot always increase (or decrease) his index (according to the simplicial ordering). Thus, there is some step where the robber moves from $v_x$ to $v_y$ to $v_z$, where $y > x$ and $y> z$. Suppose $C$ occupies $v_i$ and $R$ occupies $v_x$ where $d(v_i,v_x)=\ell$. We will now show that $C$ will eventually capture the robber.

After $R$ moves to $v_y$, $C$ moves to any neighbour of $v_i$ that is distance $\ell-1$ from $v_x$; call such a neighbour $v_j$. If $\ell =1$ then $v_x=v_j$ and the robber is captured, otherwise, $R$ moves to $v_z$.  Since the vertices are indexed by a simplicial ordering, $v_y$'s lower-indexed neighbours form a clique.  Thus $v_x$ is adjacent to $v_z$ and $d(v_j,v_z)=\ell-1$. Since $C$ now occupies a vertex that is distance $\ell-1$ from the vertex occupied by $R$, by the inductive hypothesis $C$ can now capture the robber using an $(\ell-1)$-visibility strategy.
\end{proof}

Consider the case where $\ell$ is the diameter of the graph. Then the cop can see the entire graph and will capture the robber, thus the above result provides an alternate proof that chordal graphs are cop-win.

%	   Though $c_\ell = c^\prime_\ell$ for finite chordal graphs, the result is not true in general.  Recall that the graph, $G$, given in Figure~\ref{fig:Ex1} is cop-win. We note that $G$ contains an induced 4-cycle $v_1,v_3,v_2,v_6$; thus although $G$ is cop-win, it is not chordal.% nor bridged.
%The graph $G$ shown in Figure~\ref{fig:Ex1} is used to illustrate that ``seeing does not imply capture''; i.e. if a one-visibility cop sees the robber, it does not necessarily imply the cop can capture the robber.
	  	
Theorem~\ref{thm:ellchord} raises the question: if $G$ is chordal, what is $c'_\ell(G)$? In the following section we fully answer this case for trees, a subfamily of chordal graphs. We provide a complete structural characterization for trees and show that for each $k \in \mathbb{N}$, there exists a chordal graph $G$ such that $c'_\ell(G) \geq k$.
	
Another question naturally arises, given that chordal graphs are cop-win: if $G$ is cop-win, does $c_\ell ' (G) = c_\ell (G)$?   In general, the answer is no.  In Figure~\ref{fig:Ex1}, we have a cop-win graph $G$ such that $c'_1(G) = 1$, but $c_1(G) = 2$.  Furthermore, when $\ell \ge 2$ we know that if $G$ is cop-win and $c'_\ell (G) \neq c_\ell (G)$, then $1 \le c_\ell (G) \le 2$ (from Theorem \ref{thm:see}).  It follows that for any $\ell \ge 2$ and cop-win graph $G$, $c_\ell (G) - c'_\ell (G) \le 1$.

\section{Trees}\label{sec:trees}

	Recall that for a tree $T$, the height of the tree $h(T)$ is given by $\min_{v\in V(T)}\{ec(v)\}$, where $ec(v)$ is the eccentricity of $v$.
	In \cite{Feiran}, the author gives an upper bound for the $1$-visibility cop number of a trees as a function of the height of the tree. The result is restated in Theorem \ref{upperboundtrees}, and the proof is provided for completeness.

\begin{theorem} \label{upperboundtrees}\cite{Feiran}
Given  a tree $T$,  $c_1(T) \le \left \lceil \frac{h(T)}{3} \right \rceil$.
\end{theorem}

\begin{proof}
We note that on a tree, once a cop sees the robber, the robber will be captured in a finite number of rounds (the cop simply moves onto the robber's previous position which eventually forces the robber onto a leaf).   Therefore, in the 1-visibility game, it suffices to show that at some point, a cop occupies a vertex adjacent to the robber.  We begin by verifying that for any tree with height at most three, a single cop can capture the robber.

First, suppose $T$  is a tree with height at most two. Assume $T$ is rooted at its centre (or an endpoint of its centre), $w$.  If $h(T) \le 1$, then one cop initially positioned at $w$ will capture the robber in the next round.  Suppose $h(T) = 2$ and $N(w) = \{w_1, w_2, \ldots , w_m\}$ for some $m \ge 2$.  The cop then performs the  walk $w, w_1, w, w_2, w, \ldots, w, w_m$, alternating between a neighbour of $w$ and $w$ itself.  The robber is prevented from moving onto a vertex in $N[w]$, since the cop will either see him immediately or in the next round.  Therefore, to avoid capture, the robber must occupy some vertex $x$ such that $d(w,x) = 2$ and pass in every round.  Eventually the robber will be seen by the cop.

Now, suppose $T$ is a tree rooted at its centre, $r$,  such that $h(T) = 3$.  For each $w \in N(r)$, the cop, in turn, performs the walk described in the previous paragraph.  If the robber is in the subtree rooted at $w$, he will be seen by the cop.  Furthermore, if the robber moves onto $r$, he will either be seen by the cop immediately, or he will be seen by the cop in the next round when the cop moves onto some $w \in N(r)$.   Therefore, if the robber is in a subtree of $T$ rooted at some $w \in N(r)$ in the initial round of the game, he cannot move out of that subtree without being seen by the cop.  It follows that the robber will be seen by the cop after a finite number of rounds.

Now assume that for any rooted tree $T'$ of height at most $3k$ for some $k \ge 1$, $c_1(T') \le k$.
Consider a tree $T$ rooted at its centre $r$, such that $h(T) \le 3k+3$.  Let  $N(r) = \{x_1, x_2, \ldots , x_m\}$.  For each $x \in V(T)$, let $T_x$ be the subtree of $T$ rooted at $x$.  It follows that if $d(r,x) = 3$, then $h(T_x) \le 3k$ and $c_1(T_x) \le k$.

For each $i=1, \ldots , m$ we let $X_i = \{x \, | \, d(r,x) = 3 \text{ and }  d(x_i,x) = 2\}$.  It follows that the vertices in $X_i$ are descendants of both $r$ and $x_i$.   Without loss of generality, assume that for some $m' \ge 1$, $X_1, \ldots X_{m'}$ are non-empty, and either $m' = m$ or $X_j = \emptyset$ for $j = m'+1, \ldots , m$.

The cops' strategy is as follows:  a cop, $C_1$, initially occupies $x_1$, while a set of $k$ cops, $\mathcal C$,  occupy vertices in the subtree $T_x$ for some $x \in X_1$.  The cop $C_1$ vibrates between $x_1$ and the parent of $x$ in $T$.  Meanwhile, the set $\mathcal C$ of $k$ cops clean $T_x$.  Once $T_x$ is cleaned, assuming the robber has not been seen, we mark vertex $x$.  (All vertices in $X_1$ are initially unmarked.)  We then iteratively choose an unmarked vertex $y$ in $X_1$ and move the cops in $\mathcal C$ to $y$.  Only then does $C_1$ change its vibrating pattern and begin to vibrate between $x_1$ and the parent of $y$.  (Note that $x$ and $y$ may have the same parent.)  The cops in $\mathcal C$ then clean $T_y$.  This is repeated until all vertices in $X_1$ are marked or the robber is seen.  Since the robber will be seen if he moves onto any neighbour of $x_1$, once $T_y$ is cleaned for some $y \in X_1$, it cannot be recontaminated during this phase.  As a result, the subtree rooted at $x_1$ is cleaned.
Furthermore, if the robber ever moves onto the root $r$ of $T$, he will either be seen immediately by $C$ or will been seen the next round when $C_1$ moves onto $x_1$.

% \me{With respect to the above paragraph: why does $C_1$ need to vibrate?  Can't he sit at $x_i$ while the subtrees of $x_i$ are cleaned?  By sitting at $x_i$, he would guard $r$ and the children of $x_i$.  Since $h(T_x) \leq 3k$ for $x \in X_i$ we know $c_1(T_x) \leq k$.  So one guard sits at $x_i$ and $k$ guards clean the subtrees of $x_i$ one at a time.} \dc{is it so that R can't sneak by $x$ and go back into another branch of $T_x$ that was cleaned? - or I may have missed something.}
%\cd{I think ME is correct. I see  no need for $C_1$ to vibrate. Though I will note that if we let $C_1$ vibrate, we can improve the $3$ to be a $4$ in this proof. That said, having $C_1$ vibrate does not seem to affect the proof. Given that this is a reproduction from \cite{Feiran}, we can just leave it as is?}

For each $i = 1, 2, \ldots m-1$, once the subtree $T_{x_i}$ is cleaned, $C_1$ moves on the path $x_irx_{i+1}$ and the  cops in $\mathcal C$ move onto the subtree $T_{x_{i+1}}$.  Together they then clean
$T_{x_{i+1}}$.  (If $i+1 \le m'$, they use the strategy in the previous paragraph.  Otherwise, the subtree $T_{x_{i+1}}$ has height at most one and is cleaned by $C_1$.) We note that since $C_1$ is adjacent to $r$ at some point in every round, the robber cannot move onto $r$ without being captured.  As as result once $T_{x_i}$ is cleaned, it cannot be recontaminated.  It follows that the robber will eventually be seen.
\end{proof}\\

	By noticing that in a tree of height $h\leq 2\ell+1$ the robber may never pass through the root unobserved we arrive at the following.
	
	\begin{observation}
		\label{obs3height}
		If  $T$ is a tree with height $h\leq 2\ell+1$, $\ell\geq 0$ with respect to root $r$, then $c_{\ell}(T)=1$.
	\end{observation}
	
	%This observation follows because in such a tree, the robber may never pass through the root unobserved.
	
	In \cite{Complexity}, the authors provide a full characterization for the zero-visibility cop number of trees. Here we extend this work for $\ell > 0$.

	 %Recall that $h(T)$ is the height of a tree $T$ minimized over all roots.

	%\begin{theorem}\cite{Feiran}
	%	If $T$ is a tree then  $c_1(T) \le \left\lceil \frac {h(T)}{3} \right\rceil$. 	\end{theorem}

Let $T$ be a tree with a root $r$. Let $x,y$ be vertices of $T$ where $x$ is a successor of $y$ with respect to $r$. The \emph{subtree rooted at $x$} is the component of $T - xy$ that contains $x$.	We call a vertex at distance $q$ from $r$ a \emph{$q$-descendent} of $r$.

\begin{figure}
	\begin{center}
		\includegraphics[width = 0.6\linewidth]{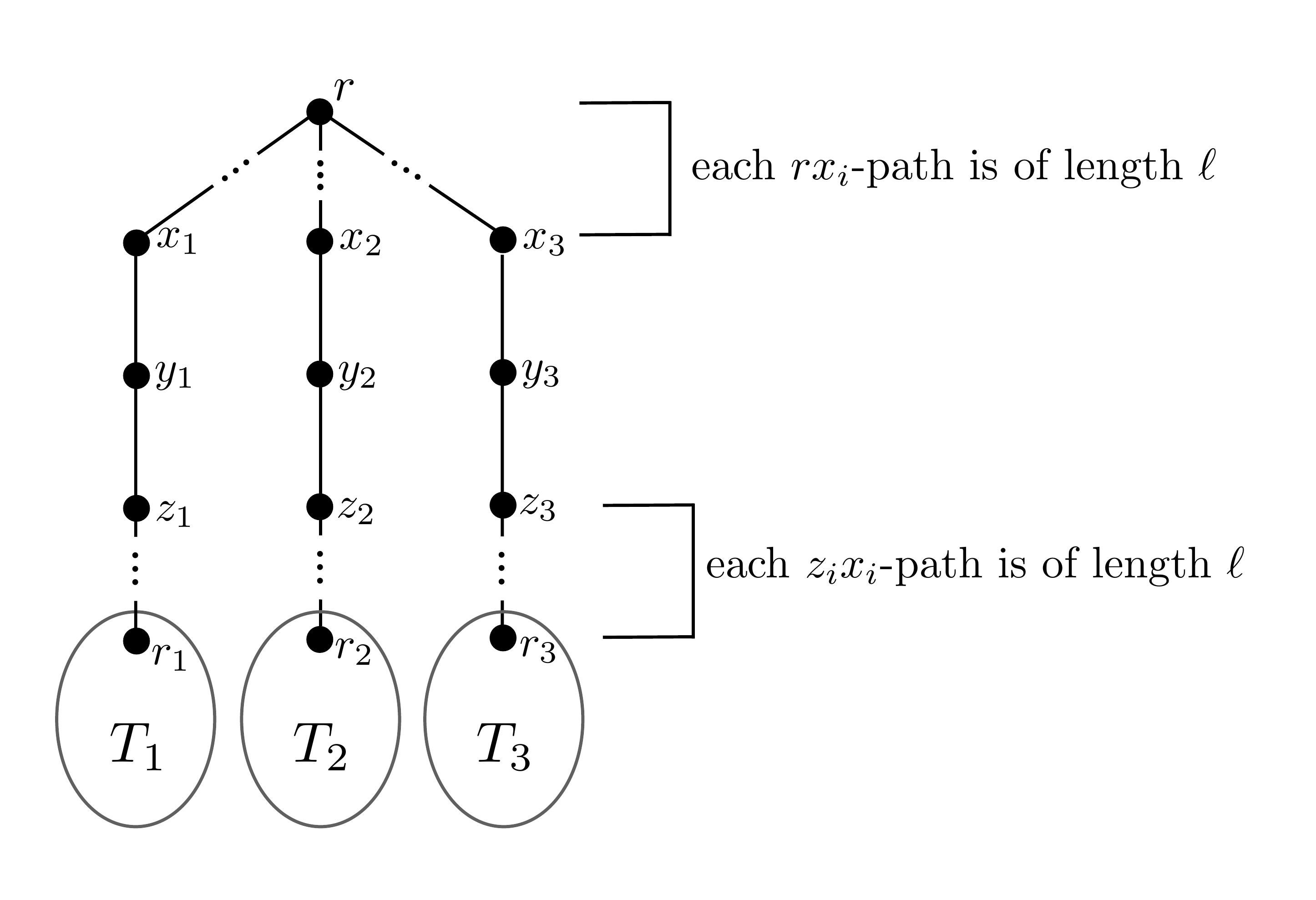}
	\end{center}
	\caption{An element of $\mathcal{T}_{k,\ell}$}
	\label{fig:kcoptrees}
\end{figure}

\begin{defn}\label{def}
	Let $\mathcal{T}_{k,\ell}$ be the family of trees defined in the following way:
	\begin{itemize}
		\item $\mathcal{T}_{1,\ell} = \{ K_1\}$
		
		\item $\mathcal{T}_{k,\ell}$ is the set of trees, $T$, that can be formed as follows: let $T_1$, $T_2$, $T_3 \in \mathcal{T}_{k-1,\ell}$. Let $r_1,r_2,r_3$ be vertices of $T_1, T_2, T_3$ respectively. Then $T$ is formed from the disjoint union of $T_1,T_2,T_3$, together with paths of length $2{\ell}+2$ from each of $r_1, r_2, r_3$, to a common endpoint, $q$.
	\end{itemize}
\end{defn}

Note that $r_1,r_2,r_3$ are any vertices of $T_1,T_2,  T_3$, respectively. They need not be the vertices $q$ identified in the construction of elements of $T_{k-1,\ell}$.

%\begin{prop}
%%If $T \in \mathcal{T}_k$ then $c_1(T) = k$.
%\end{prop}
%\begin{proof}
%Let $T \in \mathcal{T}_k$. Place $k$ cops on $r$. One cop will move between two vertices so that the root $r$ is within sight every other timestep and the remaining $k-1$ cops will clean the $T_{k-1}$. Then then move onto clean the other two paths in a similar way.

%We now show we can not use only $k-1$ cops. WLOG suppose $T_1$ has been cleaned. When we move all $k-1$ cops into the are cleaning one of the $T_{k-1}$'s a robber could move past the root, $r$ into $T_1$, therefore reconatimating it.
%\end{proof}

\begin{lemma}
	\label{lem:copnumberk}
	If $T\in \mathcal{T}_{k,\ell}$ then $c_{\ell}(T)=k$.

\end{lemma}

\begin{proof}
Let $T \in \mathcal{T}_{k,\ell}$ with vertices labelled as in Figure \ref{fig:kcoptrees}, with paths of length $\ell$ from $r$ to $x_i$ and from $z_i$ to $r_i$. Observe that the subtrees $T_1$, $T_2$, and $T_3$ are in $\mathcal{T}_{k-1,\ell}$. Inductively, $c_\ell(T_i) = k-1$ for $1 \le i \le 3$.  We first assume that $k-1$ cops suffice to clean $T$.

Let $T_i^\prime$ be the subtree rooted at $z_i$. We now show that for all $i$ there exists a round during which all $k-1$ cops are located in $T'_i$. Assume, for some cleaning strategy $\cal{C}$ for $T$ using $k-1$ cops, that in every round, there are at most $k-2$ cops in $T'_i$. Let $w_x(t)$ be the vertex occupied by cop $C_x$ at the end of round $t$ in $\cal C$ for $1 \leq x \le k-1$. Define $w'_x(t)$ as follows.

$$w'_x(t)  =  \left\{\begin{array}{cl}
w_x(t) & \mbox{ if $w_x(t) \in V(T_i)$,}\\
r_i & \mbox{ if } w_x(t) \in V(T'_i)\backslash V(T_i)\\
y_i & \mbox{ otherwise.}
\end{array}\right.$$

In this strategy, $k-2$ cops suffice to clean $T_i$, a contradiction. For each $i$ there exists a round so that all $k-1$ cops are located in $T'_i$.

Let $t_1$ be the earliest round for which exactly one leaf is cleaned at the end of round $t_1$ and this leaf remains clean during each subsequent round. Let $v$ be this leaf and without loss of generality, let it be in $T_1$.

By the above argument, there exists a round during which all $k-1$ cops are in $T^\prime_2$ and there is a different round during which all $k-1$ cops are in $T^\prime_3$. We note that these times must occur after $t_1$, otherwise by the argument above, we have a contradiction, as $k-2$ cops suffice to clean $T_2$ or $T_3$.  Of $T_2'$ and $T_3'$, assume without loss of generality that $T^\prime_3$ is the first to contain all $k-1$ cops and this occurs at the end of round $t_{3}>t_1$.

Let $t_2$ be the round after $t_1$ and before $t_3$ that all cops are on the path from $r$ to $r_3$ or are in $T_3$. Assume, without loss of generality, that from round $t_2$ to round $t_3$ that the robber maintains a distance of $\ell+2$ from the nearest cop. Then, at the end of round $t_3$, all cops are in $T_3'$ and the robber may be at $r$.  At the end of each subsequent round, the robber maintains (at least) distance $\ell+2$ from the nearest cop and will eventually recontaminate $v \in T_1$, which yields a contradiction. Therefore $k-1$ cops do not suffice to clean $T$.

We now construct a $k$-cop cleaning strategy for $T_{k}$. 	
Place all of the cops on $x_1$. One cop, $C$, will vibrate between $x_1$ and $y_1$. The remaining $k-1$ cops traverse down the branch and use a $k-1$-cop strategy to clean $T_1$. Note that $r_1$ is not necessarily the root of $T_1$, but as the $(k)$th cop prevents the robber from leaving $T_1$ unobserved, the remaining $k-1$ cops may re-position themselves as needed. Since $C$ is on $y_1$ in every second time-step, this ensures that once $T_1$ is cleaned using the $k-1$-cop strategy it remains clean.

The cops then all move to $x_2$ and repeat this process with $T_2$. Since $C$ is on $x_2$ in every second time step while $T_2$ is being cleaned, the robber cannot pass through $r$ unobserved. Similarly when repeating this process for $T_3$. Hence we have cleaned $T$ using $k$ cops.
\end{proof}

\begin{corollary}\label{cor:cleanRoot}
	For all $T \in \mathcal{T}_{k,\ell}$ there exists a cleaning strategy for $T$ using $k$ cops in which the root of $T$ is never unseen for two consecutive timesteps.
\end{corollary}

\begin{lemma} \label{lem:orange}
	Let $T$ be a tree with root $r$ such that each of the subtrees rooted at an $\left(\left\lceil\frac{i}{2}\right\rceil\right)$-descendant have $\ell$-visibility cop number at most $k-2$, for some $i\leq 2\ell+2$. There exists a cleaning strategy for $T$ using $k-1$ cops such that $r$ is occupied every second time-step.
\end{lemma}

\begin{proof}
	Let $x$ be a neighbour of $r$. Let $S_x$ be the set of $\left(\left\lceil\frac{i}{2}\right\rceil\right)$-descendant of $r$ that is also a descendants of $x$. If a single cop $C$ vibrates between $x$ and $r$, and if $S_x \neq \emptyset$, then the remaining $k-2$ cops clean the subtrees rooted at elements of $S_x$. If $S_x = \emptyset$, then moving $C$ to $x$ suffices to clean the subtree rooted at $x$.
	Repeating this process for all neighbours of $r$ gives the required cleaning strategy.
\end{proof}

\begin{lemma} \label{lem:white}
	Let $T$ be a tree with root $r$ such that each of the subtrees rooted at an ${i}$-descendant of $r$ have $\ell$-visibility cop number at most $k-2$ for some $i< 2\ell+2$. There exists a cleaning strategy for $T$ using $k-1$ cops such that $r$ is seen every second time-step.
\end{lemma}

\begin{proof}
	Let $S_r$ be the set of $1$-descendants of $r$. For $x \in S_r$ let $S_x$ be the $\left(\left\lfloor\frac{i}{2}\right\rfloor\right)$-descendants of $r$ that are descendants of $x$.
	By Lemma \ref{lem:orange}, a subtree rooted at an element of $S_x$ can be cleaned using $k-1$ cops such that $r$ is seen in every second time-step.
	As such we can clean the subtree rooted at $x$ using $k-1$ cops by cleaning each of the subtrees rooted at an element of $S_x$ in left to right order such that $r$ is seen in every second time-step.
	Observe that any vertex that is a descendant of $x$ that is not contained in a subtree rooted at an element of $S_x$ is cleaned during this process.
	Repeating this process for each of the elements of $S_r$ gives a cleaning strategy for $T$ using $k-1$ cops such that $r$ is seen in every second time-step.
	Note that since $r$ is seen in every second time-step, once a subtree rooted at an element of $S_r$ is clean it cannot be recontaminated. 	
\end{proof}

\begin{theorem}\label{thm:treek}
	If $T$ is a tree, then $c_{\ell}(T)=k$, where $k$ the greatest integer such that $T$ contains a subgraph from $\mathcal{T}_{k,\ell}$.
\end{theorem}

\begin{proof} Let $T$ be a tree and $k$ the greatest integer such that $T$ contains a subgraph from $\mathcal{T}_{k,\ell}$. Clearly $T$ needs at least $k$ cops by Lemma~\ref{lem:copnumberk} and Corollary~\ref{cor:isometrictree}. It suffices to show that if $c_{\ell}(T)=k$, then $T$ must contain some $T_k\in \mathcal{T}_{k,\ell}$.

Now, let $k$ be the least integer such that there exists a tree $T$ such that $c_{\ell}(T)\geq k$ and does not contain some element of $\mathcal{T}_{k,\ell}$. Since $k$ is the least such integer, $T$ contains an element of $\mathcal{T}_{k-1,\ell}$. We proceed in two cases.\vspace{0.1in}

%\vspace{0.1in}\noindent\underline{PROOF FOR $\ell$}:

\noindent\emph{Case I: $T$ contains no pair of vertex disjoint  elements of $\mathcal{T}_{k-1,\ell}$}.\vspace{0.1in}

Let $H$ be any minimal subtree of $T$ such that $c_\ell(H) = k-1$ and $T \setminus H$ is connected.
By minimality of $k$, $H$ has as a subgraph an element of $\mathcal{T}_{k-1,\ell}$. By construction of $H$, there is a single vertex $x_1$ of  such that $x_1$ has a single neighbour not in $H$.
Let $x_2$ be this neighbour.
Since $H$ is minimum, each of the subtrees rooted at a $1$-descendant of $x_1$ that are contained in $H$ require at most $k-2$ cops.
By assumption, $T \setminus H$ contains no element of $\mathcal{T}_{k-1,\ell}$.
Therefore, by the minimality of $k$, $c_\ell(T \setminus H) < k-1$. Thus $T$ can be cleaned using $k-1$ cops by leaving a single cop on $x_1$ and cleaning all of the subtrees of $T \setminus\{x_1\}$ one at a time, each using $k-2$ cops.\vspace{0.1in}

\noindent\emph{Case II: $T$ contains at least two elements of $\mathcal{T}_{k-1,\ell}$ that share no vertices in $T$.}\vspace{0.1in}

Let $H$ and $H^\prime$ be minimal  vertex disjoint  subtrees of $T$ such that $c_\ell(H) =c_\ell(H^\prime) = k-1$  and each of $T \setminus H$ and $T \setminus H^\prime$ is connected. We choose $H$ and $H^\prime$ to be at maximum distance among all possible pairs of minimal vertex disjoint subtrees of $T$ that satisfy these properties.

Let $P = x_1, x_2, \dots, x_d$ be the path from $H$ to $H^\prime$ such that only $x_1$ and $x_d$ are contained in $H$ and $H^\prime$, respectively. (See Figure \ref{fig:spine}.)  We note that we need not have $x_1$ or $x_d$ be the vertex identified as $q$ in the construction of elements of $\mathcal{T}_{k-1,\ell}$.
Let $r$ be a centre vertex this path.
We consider $T$ to be rooted at $r$.
For $2 \leq i \leq d-1$, let $X_i$ be the subgraph containing $x_i$ induced by the deletion of $E(P)$ from $T$.
We  show  $c_\ell(T) = k-1$.

By the minimality of $H$, each of the subtrees of $H\setminus\{x_1\}$, say $H_1, H_2$, \dots $H_t$, require at most $k-2$ cops to be cleaned. As such,  $H$ can be cleaned using $k-1$ cops by leaving a single cop on $x_1$ and using the remaining $k-2$ cops to clean, in order, each $H_a$ $(1 \leq a \leq t)$. Similarly, if every vertex excluding those in $H^\prime$ are clean and all of the cops are at $x_d$, we can clean the vertices of $H^\prime$ so that no vertex of $T \setminus H^\prime$ becomes recontaminated to complete the cleaning of $T$ with $k-1$ cops.
As such we may assume $d \geq 3$.
To complete the proof it suffices to show that $X_i$ can be cleaned using at most $k-1$ cops so that $x_i$ is seen by a cop in at most every second time-step for $1 < i < d$.

By Lemma \ref{lem:white} each of $X_1, \dots X_t$ can be cleaned so that $x_i$ $(1 \leq i \leq t)$ is seen in every second time-step for $t < \min \{2\ell + 2, \lfloor\frac{d}{2}\rfloor \}$.

  If $\frac{d}{2} \leq 2\ell + 2$, then observe that for  $\frac{d}{2} \leq i \leq  2\ell + 2$  a subtree rooted at a $i$-descendent of $x_i$ cannot require $k-1$ cops to be cleaned, otherwise such a subtree, $T''$ would contain an element of $\mathcal{T}_{k-1,\ell}$ with $q = x_i$, $T_1 = H$, $T_2 = H^\prime$, $T_3 = T''$. Thus, by Lemma \ref{lem:white}, each $X_i$ can  be cleaned so that $x_i$  is seen in every second time-step for $i \leq \lfloor\frac{d}{2}\rfloor$.

  By symmetry, each of the trees $X_{\lfloor\frac{d}{2}\rfloor}, \dots X_{d-1}$ can be cleaned by applying Lemma \ref{lem:white}. This completes the proof. \end{proof}

%  	--------------------------------------
%
%
%
%\vspace{0.1in}\noindent\underline{PROOF FOR $\ell = 1$}:\\
%
%\noindent\emph{Case I: $T$ contains no pair of vertex disjoint  elements of $\mathcal{T}_{k-1,\ell}$}.\vspace{0.1in}
%
%Let $H$ be any minimal subtree of $T$ such that $c_\ell(H) = k-1$ and $T \setminus H$ is connected.
%By minimality of $k$, $H$ has as a subgraph an element of $\mathcal{T}_{k-1,\ell}$. By construction of $H$, there is a single vertex $x_1$ of  such that $x_1$ has a single neighbour not in $H$.
%Let $x_2$ be this neighbour.
%Since $H$ is minimum, each of the subtrees rooted at a $1$-descendant of $x_1$ that are contained in $H$ require at most $k-2$ cops.
%By assumption, $T \setminus H$ contains no element of $\mathcal{T}_{k-1,\ell}$.
%Therefore, by the minimality of $k$, $c_\ell(T \setminus H) < k-1$. Thus $T$ can be cleaned using $k-1$ cops by leaving a single cop on $x_1$ and cleaning all of the subtrees of $T \setminus\{x_1\}$ one at a time, each using $k-2$ cops.

\begin{figure}[htbp]		
	\[ \includegraphics[width=\textwidth]{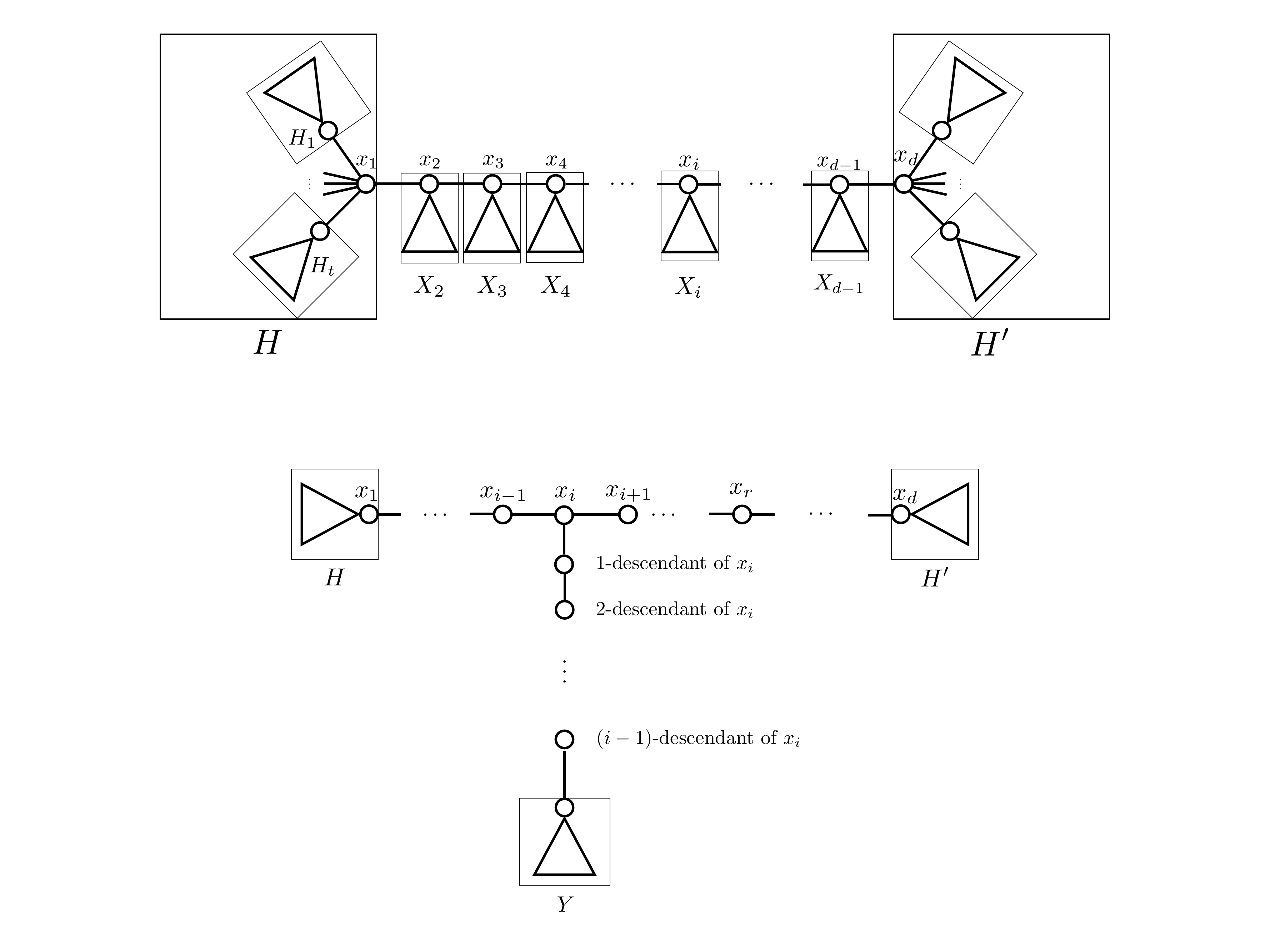} \]
	\caption{$T$ in Case II of  the proof of Theorem \ref{thm:treek}.}		
	\label{fig:spine}
	
\end{figure}

\begin{corollary}
	Given  a tree $T$,  $c_\ell(T) \leq \left \lceil \frac{h(T)}{2\ell + 2} \right \rceil$.
\end{corollary}

\begin{proof}  Observe that every tree contained in $\mathcal{T}_{k,\ell}$ has height at least $k\cdot(2\ell + 2)$.
	Therefore a tree $T$ of height $h(T)$ contains no element of $\mathcal{T}_{\left \lceil \frac{h(T)}{2\ell + 2} \right \rceil +1,\ell}$.	
	By Theorem \ref{thm:treek} we have  $c_\ell(T) \leq \left \lceil \frac{h(T)}{2\ell + 2} \right \rceil$.
\end{proof}

\section{Open Problems} \label{sec:open}
In the original game of Cops and Robber, cops can see the location of the robber throughout the game.  In the other extreme, in zero-visibility Cops and Robber, cops cannot see the location of the robber unless a cop occupies the same vertex as the robber.  The more general $\ell$-visibility game covers the spectrum in between the original Cops and Robber and zero-visibility Cops and Robber and we have seen this problem to be distinct from the two extremes.  For example, Theorem~\ref{thm:treek} showed that for any choice of $\ell \ge 0$ and $k \ge 1$, there exists a tree such that $c(T) = 1$ and $c_\ell (T) = k$.  Naturally then, we ask for a graph $G$ with inequality amongst the cop numbers across the visibility spectrum.   We note that when $\ell$ is the radius of $G$, $c_\ell (G) \le c(G) +1$, since one cop sitting on a centre vertex can maintain sight of the robber while the $c(G)$ other cops capture the robber.  Therefore, increasing $\ell$ past the radius of $G$ will lead to $c_\ell (G)$ decreasing in value.  We therefore, pose the following problem regarding visibility across the spectrum.

\begin{problem} Provide a graph $G$ for which $c_0(G) > c_1(G) > \dots > c_{rad (G)}(G) > c(G)$. \end{problem}

In Section~\ref{sec:see}, we show that for $\ell \geq 1$, if an $\ell$-visibility cop sees the location of a robber on a chordal graph, then that cop can capture the robber.  Thus, although $c_\ell'(G)$ cops are needed to see the location of the robber; after this point, only one $\ell$-visibility cop is necessarily in order to eventually capture the robber.   Which cop-win graphs, beyond chordal graphs (including trees), have this property?

\begin{problem} Given a graph $G$, suppose that if an $\ell$-visibility cop, $C$, is within distance $\ell$ of the robber in some round $t$, then $C$ can then capture the robber in round $t'$ for some $t' >t$, where the play of the other cops in rounds $t+1$ through $t'$ is irrelevant.    Characterize such graphs.  \end{problem}

From Theorem \ref{thm:see}, we have established a relationship between $c_\ell (G) - c'_\ell (G)$ and $ c_\ell (G) - c(G)$, whenever $\ell \ge 2$: if $c_\ell (G) - c'_\ell (G) >0$, then $c_\ell (G) - c(G) \le 1$.  We have also established that when $G$ is chordal $c_\ell (G) - c'_\ell (G) =0$ for any $\ell \ge 0$.  However, in general, we do not currently have bounds on these differences.
\begin{problem}  For each $\ell \ge 1$, are the following differences bounded?
\begin{enumerate}
\item  $c_\ell (G) - c'_\ell (G)$
\item $c_\ell (G) - c(G)$
\end{enumerate}

\end{problem}

Because ``seeing implies capture" on chordal graphs, we are interested in the number of $\ell$-visibility cops required to see the location of the robber, $c_\ell'$.  However, $c(G)$ is not a lower bound for this parameter, as evidenced by Proposition~\ref{cor:complete}:  for example $c_2'(C_6) =1$ and yet $c(C_6) = c_2(C_6) = 2$.

\begin{problem} For $\ell \geq 1$, find a lower bound for $c_\ell'(G)$.  \end{problem}

\begin{problem} Suppose there is a cost associated with increasing the visibility of the cops. We then want to consider the ratio $\frac{c_{\ell}(G)}{c_{\ell+1}(G)}$.  What is the closure of $\frac{c_{\ell}(G)}{c_{\ell+1}(G)}$? \end{problem}

%\begin{question} Is there a relationship between $c_t(G)$, $c_1(G)$, and $c_1'(G)$? \end{question}

In~\cite{Complexity}, the authors provide a non-trivial family of graphs $\mathcal{G}_n$ and show that given $G \in \mathcal{G}_n$ and an integer $k > 0$, the problem of deciding whether $c_0(G) \leq k$ is NP-complete.  To do this, the authors relate $c_0(G)$ to the problem of computing pathwidth of $G$, which is known to be NP-complete \cite{pwcomplexity}.  By restriction this gives directly that given a graph $G$, and integers $\ell\geq 0$ and $k > 0$ that the problem of deciding if $c_\ell(G) \leq k$ is NP-complete. However, as pathwidth can be computed in polynomial time for several large classes of graphs, it gives hope that there are families of graphs for which the $\ell$-visibility cop number may be efficiently computed.

We conclude with a variant of Cops and Robbers.  In Section~\ref{intro}, we observe that once a $1$-visibility cop, $C$ has seen the robber (i.e. has occupied a vertex adjacent to the vertex occupied by the robber), $C$ can ``tail'' the robber, by following the path of vertices previously occupied by the robber.  Then $C$ can see the location of the robber at each subsequent round -- but of course this strategy of illuminating previous locations of the robber may not lead to capture.  It does however, lead us to a {\it time-delayed} variant of Cops and Robber.  During round $0$, cops choose a set of vertices to occupy and then the invisible robber chooses a vertex to occupy.  During each round $t > 0$, the cops move and then the robber moves and then the vertex occupied by the robber at the end of round $t-1$ is illuminated for the cops.  Thus, at the end of each round, the cops see the previous location of the robber (if the robber chooses to stay at his current vertex, the previous and current location of the robber may coincide).  This time-delay variant emulates the situation when an $\ell$-visibility cop occupies a vertex distance $\ell$ from the vertex occupied by the robber and follows the ``tailing'' strategy outlined above.  Denote by $c_t(G)$, the minimum number of cops required to capture the robber on a graph $G$.   Obviously $c(G) \leq c_t(G) \leq \gamma(G)$ for all graphs $G$, so the parameter is well-defined, and it is easy to find graphs for which strict inequality holds.

Throughout we have assumed that all graphs are finite, simple and contain a single robber. However, following work on Cops and Robber for infinite graphs \cite{infinite} it is possible to consider $\ell$-visibility Cops and Robber on infinite graphs. The result of Theorem \ref{thm:ellchord} does not extend to infinite chordal graphs with two-way infinite paths. In \cite{bonatoinf}  the authors show that the cop-number of the Rado graph is $\aleph_0$.   By \cite{erdos}, the same is true for almost all countably infinite graphs. However, there do exist infinite graphs for which the cop-number is finite. A first question in the study of $\ell$-visibility Cops and Robber on infinite graphs is to classify those infinite graphs $G$ and those integers $\ell$ such that $c_\ell(G) = \mathfrak{N}_0$, but $c_{\ell+1}(G)$ is finite.

\section{Acknowledgements}
N.E.~Clarke acknowledges research support from NSERC (2015-06258). D.~Cox acknowledges research support from NSERC and Mount Saint Vincent University. C.~Duffy acknowledges research support from AARMS. D.~Dyer acknowledges research support from NSERC. M.E.~Messinger acknowledges research support from NSERC (grant application 356119-2011) and Mount Allison University.

%%%%%%%%%%%%%%%%%%%%%%%%%%%%%%%%%%%%%%%%%%%%%%%%%%%%%%%%%%%%%%%%%%


\begin{thebibliography}{99}

\bibitem{AFGP16} T.V.~Abramovskaya, F.V.~Fomin, P.A.~Golovach, and M.~Pilipczuk, How to hunt an invisible rabbit on a graph. {\it  Eur. J. Comb.} {\bf 52}  (2016)12--26.

	\bibitem{pwcomplexity}  S.~Arnborg, D.G.~Corneil, and A.~Proskurowski, Complexity of finding embeddings in a $k$-tree. {\it SIAM J. on Algebraic and Discrete Methods} {\bf 8}(2) (1987) 277--284.
	
	\bibitem{BI} A.~Berarducci and B.~Intrigila, On the cop number of a graph. {\it Adv. in Appl. Math.} {\bf 14} (1993) 389--403.
	
\bibitem{BiSe91} D. Bienstock and P. Seymour, Monotonicity in graph searching. {\it J. Algorithms} {\bf 12}
 (1991) 239--245.
	
\bibitem{BGHK09} A.~Bonato, P.~Golovach, G.~Hahn, and J.~Kratochv\'{i}l, The capture time of a graph. {\it Discrete Math.} {\bf 309} (18) (2009) 5588--5595.

	\bibitem{bonatoinf} A.~Bonato, G.~Hahn, and C.~Wang, The cop density of a graph.  {\it Contrib. Discrete Math.} {\bf2}2 (2007) 133--144.

\bibitem{Zombie1} A. Bonato, D. Mitsche, X. P\'{e}rez-Gim\'{e}nez, P. Pra{\l}at, A probabilistic version of the game of zombies and survivors on graphs. {\it Theoretical Computer Science} {\bf 655} (2016) 2--14.
	
\bibitem{Cops} A.~Bonato and R.J.~Nowakowski, {\it The Game of Cops and Robbers on Graphs}, American Mathematical Society, Providence, Rhode Island, 2011.
	
\bibitem{CM12} N.E.~Clarke and G.~MacGillivray, Characterizations of $k$-cop-win graphs. {\it Discrete Math.} {\bf 312}(8) (2012) 1421--1425.

\bibitem{CN05} N.E.~Clarke and R.J.~Nowakowski, Tandem-win graphs. {\it Discrete Math.} {\bf 299} (2005) 56--64.

	\bibitem{Pathwidth} D.~Dereniowski, D.~Dyer, R.M.~Tifenbach, and B.~Yang, Zero-visibility cops and robber and the pathwidth of a graph. {\it J. Comb. Optim.} {\bf 29}(3) (2015) 541--564.
	
	\bibitem{Complexity} D.~Dereniowski, D.~Dyer, R.M.~Tifenbach, and B.~Yang, The complexity of zero-visibility cops and robber. {\it Theoret. Comput. Sci.} {\bf 607} (2015) 135--148.
	
	\bibitem{erdos} P.~Erd\"os and  A.~R\`enyi, Asymmetric graphs. {\it Acta Math. Acad. Sci. Hungar} {\bf 14} (1963) 295--315.
	

	\bibitem{Feiran} F.~Yang,  1-visibility Cops and Robber Problem, Honour's thesis, University of Prince Edward Island, 2012.
	
	\bibitem{Zombie} S.L.~Fitzpatrick, J.~Howell, M.E.~Messinger, and D.A.~Pike, A deterministic version of the game of zombies and survivors on graphs. {\it Discrete Appl. Math.} {\bf 213} (2016) 1-12.
	
	\bibitem{infinite} G.~Hahn, F.~Laviolette, N.~Sauer, and R.E.~Woodrow, On cop-win graphs. {\it Discrete Math.}, {\bf 258}, 27-41.
	
\bibitem{KMP13} A.~Kehagias, D.~Mitsche, and  P.~Pra\l{}at, Cops and invisible robbers: The cost of drunkenness. {\it Theoret. Comp. Sci.} {\bf 481} (2013) 100--120.

\bibitem{LaPa93} A.S.~LaPaugh, Recontamination does not help to search a graph. {\it J.  Assoc. Comput. Mach.} {\bf 40} (1993) 224--245.

	\bibitem{NR}  R.~Nowakowski and I.~Rival, On a class of isometric subgraphs of a graph, {\it Combinatorica} {\bf 2}(1) (1982) 79--90.
	
	\bibitem{NW} R.J.~Nowakowski and P.~Winkler, Vertex-to-vertex pursuit in a graph, {\it Discrete Math.} {\bf 43} (1983) 235--239.
	

\bibitem{bushcutting} P.~Gordinowicz, {\it Those Magnificent Blind Cops in Their Flying Machines with Sonars}, at GRASTA 2017.

	\bibitem{Q} A.~Quillot, {\it Th\`{e}se d'Etat}, Universit\'{e} de Paris VI, 1983.

	\bibitem{Tang} A.~Tang, Cops and robber with bounded visibility, Master's thesis, Dalhousie University, 2004.
	
	\bibitem{Tosic} R.~To\u{s}i\'{c}, Vertex-to-vertex search in a graph, Graph Theory (Dubrovnik 1985) (1985) 233--237.
	
	\bibitem{west} D.B.~West, {\it Introduction to Graph Theory}, 2nd Ed., Prentice-Hall, New Jersey, 2001.
	
\bibitem{YaDyAl09}	B.~Yang, D.~Dyer, and B.~Alspach, Sweeping graphs with large clique number. {\it Discrete Math.} {\bf 309} (2009), no. 18, 5770--5780.


	

	
\end{thebibliography}
\end{document}